\documentclass[11pt]{article} 

\usepackage[utf8]{inputenc}
\usepackage[T1]{fontenc}
\usepackage[pdftex,left=1in,top=1in,bottom=1in,right=1in]{geometry}

\usepackage{amsmath}
\usepackage{amssymb}
\usepackage{dsfont}
\usepackage{mathrsfs}
\usepackage{amsthm}

\usepackage{graphicx}
\usepackage{caption}
\usepackage{bm}
\usepackage{xcolor}
\usepackage{thm-restate}
\usepackage[colorlinks=true,citecolor=blue,linkcolor=blue]{hyperref}
\hypersetup{
    colorlinks,
    linkcolor={red!50!black},
    citecolor={blue!50!black},
    urlcolor={blue!80!black}
}

\usepackage[nameinlink, noabbrev, capitalize]{cleveref}
\AddToHook{cmd/appendix/before}{\crefalias{section}{appendix}}

\newcommand{\N}{\mathbb{N}}

\newcommand{\R}{\mathbb{R}}
\newcommand{\E}{\mathds{E}}

\newcommand{\cC}{\mathcal{C}}

\newcommand{\cX}{\mathcal{X}}
\newcommand{\cY}{\mathcal{Y}}

\newcommand{\Capa}{\operatorname{Cap}}

\newcommand{\eps}{\varepsilon}

\newtheorem{theorem}{Theorem}[section]

\AddToHook{env/proposition/begin}{\crefalias{theorem}{proposition}}

\AddToHook{env/property/begin}{\crefalias{theorem}{property}}

\newtheorem{claim}[theorem]{Claim}

\newtheorem{lemma}[theorem]{Lemma}
\AddToHook{env/lemma/begin}{\crefalias{theorem}{lemma}}

\AddToHook{env/conjecture/begin}{\crefalias{theorem}{conjecture}}

\AddToHook{env/corollary/begin}{\crefalias{theorem}{corollary}}

\AddToHook{env/example/begin}{\crefalias{theorem}{example}}

\newtheorem{definition}[theorem]{Definition}
\AddToHook{env/definition/begin}{\crefalias{theorem}{definition}}

\newtheorem{remark}[theorem]{Remark}
\AddToHook{env/remark/begin}{\crefalias{theorem}{remark}}

\AddToHook{env/question/begin}{\crefalias{theorem}{question}}

\AddToHook{env/condition/begin}{\crefalias{theorem}{condition}}

\newcommand{\Bdel}{B^{\mathrm{del}}}
\newcommand{\Bins}{B^{\mathrm{ins}}}

\newcommand{\bits}{\{0,1\}}

\newcommand{\Dec}{\mathsf{Dec}}

\newcommand{\BDC}{\mathrm{BDC}}
\newcommand{\Ch}{\mathrm{Ch}}

\newcommand{\zo}{\{0,1\}}

\newcommand{\Cins}{C_{\mathrm{ins}}}

\newcommand{\Cdel}{C_{\mathrm{del}}}

\let\originalleft\left
\let\originalright\right
\renewcommand{\left}{\mathopen{}\mathclose\bgroup\originalleft}
\renewcommand{\right}{\aftergroup\egroup\originalright}

\title{The Insertion List-Decoding Capacity \\ and an Improved Bound on the Deletion List-Decoding Capacity}

\allowdisplaybreaks

\author{Roni Con\thanks{School of Electrical \& Computer Engineering, Tel Aviv
University. \texttt{ronicon@tauex.tau.ac.il}.} \and Dean Doron\thanks{Stein Faculty of Computer and Information Science, Ben-Gurion University. \texttt{deand@bgu.ac.il}.} \and João Ribeiro\thanks{Instituto de Telecomunicações and Departamento de Matemática, Instituto Superior Técnico, Universidade de Lisboa. \texttt{jribeiro@tecnico.ulisboa.pt}.}
}

\date{}

\begin{document}
	
\maketitle

\begin{abstract}
    Informally, the capacity of list-decoding in a given adversarial error model is the largest rate at which we can list-decode with list size polynomial in the block length.
    The capacity of list-decoding from insertions and deletions is a basic, yet poorly understood, aspect of coding against synchronization errors.
    For example, when dealing with a $\delta>1/2$ fraction of insertions, the best known lower bounds give little more than the fact that the capacity is positive.
    Beyond that regime we also only have loose bounds, with the lower bounds stemming from the analysis of uniformly random codes. 
    
    We make progress in our understanding of the limits of list-decoding binary codes from insertions and deletions.
    We show that the capacity of list-decoding from a $\delta$-fraction of insertions is exactly 
    \begin{equation*}
        (1+\delta)\left(1-h\left(\frac{\delta}{1+\delta}\right)\right)
    \end{equation*}
    for all $\delta\in[0,1]$, where $h$ is the binary entropy function, achieved with high probability by a code sampled according to a symmetric $2$-state Markov chain.
    Curiously, we complement this by showing that such an approach does not beat uniformly random coding in list-decoding from deletions.
    We also give an improved upper bound on the capacity of list-decoding from a $\delta$-fraction of deletions, showing in particular that it behaves like $1-h(\delta)+o(\delta)$ when $\delta\to 0$.
    This matches the first-order asymptotic behavior of the capacity of the binary deletion channel for vanishing deletion probability, but differs in the second-order term.
\end{abstract}

\newpage
\tableofcontents
\newpage

\section{Introduction}\label{sec:intro}

Errors that cause loss of synchronization, such as insertions and deletions, are a long-standing challenge in coding theory, with many basic questions still poorly understood.
On the probabilistic side, we still only know loose bounds on the capacity of channels with memoryless insertions and deletions~\cite{Mit09,CR20}, while the capacity of channels with memoryless bit flips and erasures has been known since the early work of Shannon~\cite{Sha48}.
On the adversarial side, our knowledge of uniquely-decodable and list-decodable codes for insertions and deletions remains primitive when compared to the state of the art for bit flips and erasures, in particular over small alphabets (for several perspectives on this, see the surveys of Sloane~\cite{Slo02}, Mercier, Bhargava, and Tarokh~\cite{MBT10}, and Haeupler and Shahrasbi~\cite{HS21}).

In this work, we study the optimal rate of list-decoding from insertions and deletions, also called the \emph{list-decoding capacity}, over a binary alphabet.
More precisely, for a fixed fraction of errors $\delta$, the list-decoding capacity is the largest number $R^\star$ such that for any $R<R^\star$ there exists a family of binary codes $(\cC_n)_{n\in\N}$, with each $\cC_n\subseteq\bits^n$, of asymptotic rate at least $R$ that are list-decodable against a $\delta$-fraction of errors with list size $L$ upper bounded by a polynomial in the block length $n$.

Determining the list-decoding capacity for bit flips and erasures is an easy exercise.
For example, the capacity of list-decoding from a $\delta$-fraction of erasures is $1-\delta$, achieved with high probability by a uniformly random code, and it exhibits a sharp phase transition: for any $\eps>0$ and $R\leq 1-\delta-\eps$ there is a family of codes that is list-decodable against a $\delta$-fraction of erasures with \emph{constant} list size $L=\Theta(1/\eps)$, while for $R\geq 1-\delta+\eps$ any code of rate $R$ has \emph{exponential} list size $2^{\Omega(\eps n)}$.
An analogous statement holds for the capacity of list-decoding from a $\delta$-fraction of bit flips, which equals $1-h(\delta)$ for $\delta\leq 1/2$ and $0$ for $\delta>1/2$ with $h$ the binary entropy function, again achieved with high probability by a uniformly random code.
Curiously, in both cases the list-decoding capacity coincides with the Shannon capacity of the natural corresponding probabilistic error model.
More precisely, the capacity of list-decoding from bit flips equals the capacity of the binary symmetric channel, and the capacity of list-decoding from erasures equals the capacity of the binary erasure channel.
In fact, we know that this correspondence between list-decoding capacity and channel capacity is more than a mere coincidence for some error models~\cite{PSW25}.

\paragraph{Prior work on the capacity of list-decoding binary codes from insertions and deletions.}
Our understanding of the capacity of list-decoding from insertions and deletions is much more primitive.
Prior work obtained some loose bounds on the capacity, with the lower bounds stemming from analyzing the performance of uniformly random codes.
We discuss the present state of affairs in more detail now.
We first specialize the discussion to the binary alphabet and to correcting solely deletions or solely insertions, which is our setting of interest, and later in \cref{sec:related} point out prior work that extends beyond that.

We say that a code $\cC\subseteq\bits^n$ is \emph{$(\delta,L)$-list-decodable from insertions} if for every $y\in\bits^{(1+\delta)n}$ we have
\begin{equation*}
    |\{c\in\cC:\textrm{ $y$ is a supersequence of $c$}\}|\leq L.
\end{equation*}
Analogously, we say that $\cC\subseteq\bits^n$ is \emph{$(\delta,L)$-list-decodable from deletions} if for every $y\in\bits^{(1-\delta)n}$ we have
\begin{equation*}
    |\{c\in\cC:\textrm{ $y$ is a subsequence of $c$}\}|\leq L.
\end{equation*}
While in the context of \emph{unique decoding} correcting $t$ deletions is equivalent to correcting $t$ insertions, for any number of errors $t$, this does not hold once we move to list-decoding~\cite{HSS18}.

We denote by $\Cins(\delta)$ the capacity of list-decoding from a $\delta$-fraction of insertions, and by $\Cdel(\delta)$ the capacity of list-decoding from a $\delta$-fraction of deletions.
Recall from above that $\Cins(\delta)$ is the largest real number such that for any rate $R<\Cins(\delta)$ and all sufficiently large $n$ there exists a code $\cC\subseteq\bits^n$ that is $(\delta,L)$-list-decodable from insertions with list size $L$ growing at most polynomially with $n$ (the definition of $\Cdel(\delta)$ is analogous).
The known results about the capacity of list-decoding from bit flips and erasures even leaves open the prospect of obtaining list size \emph{independent of $n$} at any rate below capacity. 

It is easy to see that $\Cins(\delta)=0$ for all $\delta\geq 1$ and $\Cdel(\delta)=0$ for all $\delta\geq 1/2$.
However, what happens for smaller $\delta$ remains a basic, yet still poorly understood, question about codes for synchronization errors.

With respect to deletions, initial work by Guruswami and Wang~\cite{GW17} showed that $\Cdel(\delta)>0$ for all $\delta\in[0,1/2]$ with an explicit construction.
This was then improved (non-constructively) by Haeupler, Shahrasbi, and Sudan~\cite[Theorems 1.3 and 1.6]{HSS18}, who showed that
\begin{equation}\label{eq:Cdel-bounds}
    1-h(\delta)\leq \Cdel(\delta) \leq 1-2\delta
\end{equation}
for all $\delta\in[0,1/2]$, where $h(p)=-p\log p-(1-p)\log(1-p)$ is the binary entropy function.
The lower bound is achieved with high probability by a uniformly random code, with constant list size.
It is also not hard to see (by noting that the number of length-$n$ supersequences of any length-$(1-\delta)n$ string is approximately $2^{n h(\delta)}$) that the $1-h(\delta)$ lower bound is optimal \emph{for uniformly random codes}, in the sense that a uniformly random code of rate $1-h(\delta)+\eps$ will have list size $2^{\Omega(\eps n)}$ with high probability.

With respect to insertions, a similar analysis of uniformly random codes gives the lower bound\footnote{A worse lower bound is given in~\cite[Theorem 1.7]{HSS18}, due to the use of a sub-optimal bound on the volume of the largest deletion ball among all $n$-bit strings. Combining their analysis with a sharp bound due to Hirschberg and Régnier~\cite{HR02} yields the lower bound on $\Cins(\delta)$ we state here. For completeness, we provide a full proof in \cref{app:unif-random}, which also shows that the lower bound cannot be improved \emph{for uniformly random codes}.}
\begin{equation}\label{eq:Cins-LB-prior}
    \Cins(\delta)\geq \begin{cases}
        1-h(\delta),& \textrm{ if $\delta\leq 1/2$,}\\
        0, & \textrm{ if $\delta>1/2$,}
    \end{cases}
\end{equation}
and a simple averaging argument (using the fact that the number of length-$(1+\delta)n$ supersequences of any length-$n$ string is approximately $2^{(1+\delta)n h\left(\frac{\delta}{1+\delta}\right)}$) gives the upper bound~\cite[Theorem 1.2]{HSS18}\footnote{The expression in~\cite[Theorem 1.2]{HSS18} differs from the one in \cref{eq:Cins-UB}, but it is easy to see that the two expressions are equivalent.}
\begin{equation}\label{eq:Cins-UB}
    \Cins(\delta)\leq (1+\delta)\left(1-h\left(\frac{\delta}{1+\delta}\right)\right)
\end{equation}
for all $\delta\in[0,1]$.
Similarly to deletions, the lower bound in \cref{eq:Cins-LB-prior} is achieved with high probability by a uniformly random code with constant list size, and cannot be improved by a direct analysis of uniformly random codes (see \cref{app:unif-random} for a more detailed discussion)..\footnote{Liu, Tjuawinata, and Xing~\cite[Corollary 13 and Figure 2]{LTX21} claim an improved lower bound on $\Cins(\delta)$ via an analysis of uniformly random codes. However, as confirmed by the authors, their claim is not correct. More precisely, their Lemma 9 is not correct, which in turn affects their Corollary 13~\cite{LTX26}.
As we discuss in more detail in \cref{app:unif-random}, the lower bound in \cref{eq:Cins-LB-prior} cannot be improved through the analysis of uniformly random codes.}

Given the lower bound in \cref{eq:Cins-LB-prior}, it is natural to wonder whether it is possible to achieve positive rate with small list size for insertion-rate $\delta>1/2$.
The fact that this holds for all $\delta\leq 0.707$ with tiny positive rate can be derived from Johnson-type bounds for codes correcting insertions and deletions applied to the code of Bukh, Guruswami, and H{\aa}stad~\cite{BGH17}.
These were developed by Wachter-Zeh~\cite{Wac18} and Hayashi and Yasunaga~\cite{HY20}, and subsequently improved by Liu, Tjuawinata, and Xing~\cite{LTX23}, and Yang~\cite{Yan26}.
The stronger result that $\Cins(\delta)>0$ \emph{for all $\delta\in[0,1]$} was proved by Guruswami, Haeupler, and Shahrasbi~\cite{GHS21} using an intricate explicit code construction.
However, their positive lower bound on $\Cins(\delta)$ is tiny, and very far from the upper bound in \cref{eq:Cins-UB}.

In sum, all we know about $\Cins(\delta)$ and $\Cdel(\delta)$ are loose bounds, and our understanding of $\Cins(\delta)$ for $\delta>1/2$ is especially poor.

\subsection{Our contributions}\label{sec:contributions}

We determine $\Cins(\delta)$ for all $\delta\in[0,1]$, obtain improved upper bounds on $\Cdel(\delta)$, and reveal a curious dichotomy regarding the performance of ``Markov-based'' random codes for list-decoding from insertions vs.\ from deletions.

\subsubsection{Exact capacity of list-decoding from insertions}\label{sec:contr-ins}

As our first contribution, we determine $\Cins(\delta)$ for all $\delta\in[0,1]$.

\begin{restatable}[{Exact capacity of list-decoding from insertions; see \cref{sec:Cins} for details}]{theorem}{cins}\label{thm:Cins-intro}
     We have
    \begin{equation*}
       \Cins(\delta) = (1+\delta)\left(1-h\left(\frac{\delta}{1+\delta}\right)\right)
    \end{equation*}
    for all $\delta\in[0,1]$.
    More precisely,
    for any $\eps>0$ and block length $n$ there exists a code $\cC\subseteq\bits^n$ of rate $\Cins(\delta)-\eps$ that is $(\delta,L=O(1/\eps))$-list-decodable from insertions, while any code of sufficiently large block length $n$ and rate $\Cins(\delta)+\eps$ will not be $(\delta,L=2^{\eps n/2})$-list-decodable from insertions.
\end{restatable}

We emphasize that our result is non-constructive.
\cref{thm:Cins-intro} improves significantly on both the lower bound based on uniformly random codes due to Haeupler, Shahrasbi, and Sudan~\cite{HSS18} for $\delta\leq 1/2$ and the tiny positive lower bound of Guruswami, Haeupler, and Shahrasbi~\cite{GHS21} for $\delta>1/2$.
Furthermore, the constant list size guaranteed at rates below $\Cins(\delta)$ also improves on the explicit construction of~\cite{GHS21}, which had list size growing very slowly with the block length $n$.

Note that the expression for $\Cins(\delta)$ in \cref{thm:Cins-intro} equals the averaging-based upper bound from \cite{HSS18} in \cref{eq:Cins-UB}.
We obtain the matching lower bound by analyzing the performance of random codes sampled \emph{with memory}.
More precisely, we consider sampling codewords independently via a symmetric order-$1$ Markov chain $X_1,X_2,\dots,X_n$ where $\Pr[X_1=1]=1/2$ and, for $i>1$, $\Pr[X_i= X_{i-1}]=\alpha$ with a carefully chosen $\alpha\in[0,1]$.

Sampling codewords with memory when dealing with synchronization errors is natural, and was inspired by prior approaches that led to improved lower bounds on the capacity of the deletion channel.
In more detail, uniformly random coding initially gave a baseline $1-h(\delta)$ lower bound on the capacity of the deletion channel with deletion probability $\delta$~\cite{DG06}, which can be significantly improved upon by considering codebooks sampled via low-order Markov chains~\cite{DG06,DM06}.
In the context of adversarial errors, codes sampled based on order-$1$ Markov chains have also been used to obtain lower bounds on the \emph{zero-error capacity threshold}\footnote{The zero-error capacity threshold for deletions is the largest $\delta^\star$ such that for all $\delta<\delta^\star$ there exists a positive-rate code that is uniquely decodable from a $\delta$-fraction of deletions. Currently, we know that $\sqrt{2}-1\leq \delta^\star<1/2$~\cite{BGH17,GHL22}.} for adversarial deletions~\cite{KMTU11}.

\subsubsection{Capacity of list-decoding from deletions}\label{sec:contr-del}

\paragraph{Markov codes for list-decoding from deletions?}
Given that the capacity of list-decoding from insertions is achieved by a code sampled via an order-$1$ Markov chain, it is natural to ask whether such an approach can improve on the state-of-the-art $1-h(\delta)$ lower bound on $\Cdel(\delta)$ obtained by considering a uniformly random code, or even determine $\Cdel(\delta)$ exactly.
Surprisingly, we show that the answer is negative: order-$1$ Markov codes never do better than uniformly random codes!
More precisely, we use large deviations techniques to show the following.\footnote{It is also natural to ask whether an improvement over uniformly random codes can be attained by combining an order-$1$ Markov code with some code expurgation techniques. We have tried a few approaches without success, and leave it as an interesting direction for future work.}

\begin{restatable}[{Markov codes are no better than random for list-decoding from deletions; see \cref{sec:Cdel-markov} for details}]{theorem}{markov}\label{thm:Cdel-markov}
    Fix a fraction of deletions $\delta\in[0,1/2]$ and $R=1-h(\delta)+\eps$.
    Suppose that $\cC\subseteq\bits^n$ is obtained by independently sampling $2^{Rn}$ codewords according to an order-$1$ Markov chain $X_1,\dots,X_n$ with each $X_i\in\bits$, $\Pr[X_1=1]=1/2$, and $\Pr[X_i=X_{i-1}]=\alpha$ for some $\alpha\in[0,1]$, and keeping repeated codewords.
    Then, with high probability, there exists $y\in\bits^{(1-\delta)n}$ such that
    \begin{equation*}
        |\{c\in\cC:\textrm{ $y$ is a subsequence of $c$}\}|\geq 2^{\Omega(\eps n)}.
    \end{equation*}
\end{restatable}

We believe that this reveals an interesting dichotomy between the problems of list-decoding from insertions and from deletions, and also between list-decoding from deletions and the capacity of the deletion channel, since order-$1$ Markov codes yield significantly improved capacity lower bounds compared to uniformly random coding.

\paragraph{Improved upper bounds on the capacity of list-decoding from deletions.}

To complement our negative result about order-$1$ Markov codes for list-decoding from deletions, we obtain improved upper bounds on $\Cdel(\delta)$.
In particular, we show that $\Cdel(\delta)=1-h(\delta)+o(\delta)$ as $\delta\to 0$.

First, we show that $\Cdel(\delta)$ is upper bounded by the Shannon capacity $\Capa(\delta)$ of the binary deletion channel with deletion probability $\delta$, which we denote by $\BDC_\delta$. 
Upper bounds on $\Capa(\delta)$ have received significant attention~\cite{DMP07,FD10,KMS10,Dal11,KM13,RD15,Che19,RC23,PR26,Pap26}, and the state-of-the-art upper bounds on $\Capa(\delta)$ improve on the previous $1-2\delta$ upper bound for $\Cdel(\delta)$~\cite{HSS18} for all $\delta\leq 0.35$.

\begin{restatable}[{Improved upper bound on $\Cdel(\delta)$}]{theorem}{BDCUB}\label{thm:capBDC-UB}
    Fix an arbitrary $\delta\in[0,1]$ and consider an arbitrary family of codes $(\cC_n)_{n\in\N}$ with each $\cC_n\subseteq\bits^n$ being $(\delta,L_n=2^{o(n)})$-list-decodable from deletions, of asymptotic rate $R$.
    Then, $R\leq \Capa(\delta)$.
\end{restatable}

We give a brief overview of the simple proof of \cref{thm:capBDC-UB} here, and leave technical details to \cref{sec:UB-Cdel}. 
We establish \cref{thm:capBDC-UB} via the following natural (previously known) high-level approach.
Fix a code $\cC\subseteq\bits^n$.
First, we observe that the list size of $\cC$ can be used to upper bound the decoding error probability of $\cC$ over the $\BDC_\delta$. Suppose that $\cC$ is $(\delta,L)$-list-decodable from deletions and consider the decoder which, given channel output $y$, produces a list of at most $L$ candidate codewords consistent with $y$, and then outputs a uniformly random codeword from this list. 
Since the $\BDC_\delta$ introduces at most $\delta n$ deletions with decent probability, this decoder errs with probability at most $1-\Omega(1/L)$.
Then, we invoke a strong converse for the $\BDC_\delta$, which states that any code of rate above capacity must have decoding error probability at least $1-2^{-\Omega(n)}$~\cite{AW71,Koz71}.\footnote{The concurrent strong converses of Ahlswede and Wolfowitz~\cite{AW71} and Kozlov~\cite{Koz71} do not give us the required lower bound on the decoding error probability. As stated, they say that the decoding error probability approaches $1$ as the block length $n$ increases, but do not quantify how fast it approaches $1$. Nevertheless, Kozlov's argument implicitly contains essentially all ingredients necessary to get this lower bound, and combining his ideas with Arimoto's later strong converse for discrete memoryless channels~\cite{Ari73} easily gives the desired result. We discuss this in more detail in \cref{sec:UB-Cdel}.}
Combining these two observations lets us conclude that $L\geq 2^{\Omega(n)}$ for every code of rate above $\Capa(\delta)$.

The simple observation above that the list size of a code can be used to control its decoding error probability over a related channel with random errors has been used before.
See, for example, the work of Pernice, Sprumont, and Wootters about list-decoding in the Hamming metric and the $q$-ary symmetric channel~\cite[Section 1.2]{PSW25}.

Our understanding of $\Capa(\delta)$ is especially solid in the low-noise regime $\delta\approx 0$.
In that setting, we know that $\Capa(\delta)=1-h(\delta)+c\delta+O(\delta^2)$, where $c\approx 0.28$~\cite{KMS10,KM13}, and \cref{thm:capBDC-UB} implies that this is also an upper bound on $\Cdel(\delta)$.
Is this the true asymptotic behavior of $\Cdel(\delta)$?

We prove an additional upper bound on $\Cdel(\delta)$ that improves on \cref{thm:capBDC-UB} in the low-noise regime and implies, together with \cref{eq:Cdel-bounds}, that $\Cdel(\delta)=1-h(\delta)+o(\delta)$ as $\delta\to 0$.
This means that, in the low-noise regime, $\Cdel(\delta)$ is smaller than $\Capa(\delta)$ by about $0.28\delta$ and smaller than $\Cins(\delta)=1-h(\delta)+\delta+O(\delta^2)$ (obtained via \cref{thm:Cins-intro}) by about $\delta$.

\begin{restatable}[{Sharper upper bound for small fraction of deletions}]{theorem}{smalldel}\label{thm:Cdel-asymp-UB-intro}
    Fix any $\varepsilon>0$ and $\delta \in (0,1/20)$. 
    Let $\cC\subseteq\bits^n$ be an arbitrary code of rate
    \[
        R \geq 1 - \delta - \left(\gamma - \delta \right) h\left( \frac{\delta}{\gamma - \delta} \right) +\eps
    \]
    with $\gamma=\frac{1}{2}-\sqrt{\frac{h(\delta)\ln 2}{2}}$.
    Then, $\cC$ has list size $2^{\Omega(\eps n)}$ from a $\delta$-fraction of deletions.
    In particular,
    \begin{equation*}
        \Cdel(\delta)\leq 1 - \delta - \left(\gamma - \delta \right) h\left( \frac{\delta}{\gamma - \delta} \right)= 1-h(\delta)+o(\delta).
    \end{equation*}
\end{restatable}

We prove \cref{thm:Cdel-asymp-UB-intro} by first showing that high-rate codes must have many codewords with many runs, and then use the fact that strings with many runs have many subsequences.
For technical details, see \cref{sec:Cdel-UB-small}.

\subsection{Additional related work}\label{sec:related}

In \cref{sec:intro} we mostly discussed work directly related to our focus on list-decoding solely from deletions or solely from insertions, for binary codes.
Here, we discuss additional related work beyond that setting.

Prior work on list-decoding from insertions and deletions has also considered non-binary alphabets and combinations of insertions and deletions, including both work on the list-decoding capacity~\cite{HSS18,GHS21,LTX21,HS22} and on Johnson-type bounds~\cite{Wac18,HY20,LTX23,Yan26}.
In particular, Guruswami, Haeupler, and Shahrasbi~\cite{GHS21} showed, with an explicit construction, that positive rate is achievable with small list size against a $\delta_{\textrm{del}}$-fraction of deletions and a $\delta_{\textrm{ins}}$-fraction of insertions if and only if $2\delta_{\textrm{del}}+\delta_{\textrm{ins}}<1$.
Improved tradeoffs between rate and error fraction were obtained by Haeupler and Shahrasbi~\cite{HS22}.
In the insertions-only ($\delta_{\textrm{del}}=0$) and deletions-only ($\delta_{\textrm{ins}}=0$) cases considered in our work the state-of-the-art bounds prior to this work remained those of \cref{eq:Cdel-bounds,eq:Cins-LB-prior,eq:Cins-UB}.

The optimal rate for \emph{unique decoding} from insertions and deletions has also received attention, and remains an important open research direction with only loose bounds known under a constant fraction of errors.
As mentioned above, we do not even know the zero-rate threshold for deletions~\cite{KMTU11,BGH17,GHL22}.
A rate lower bound was first given by Levenshtein~\cite{Lev66}, and we know that we can achieve positive rate for any $\delta<\sqrt{2}-1\approx 0.414$ fraction of deletions.
Rate upper bounds (both asymptotic and non-asymptotic) were obtained by Kulkarni and Kiyavash~\cite{KK13}. Later, Yasunaga~\cite{Yas24} obtained asymptotic upper bounds improving on a direct application of the Elias and MRRW bounds for the Hamming metric to the Levenshtein metric.

Subsequently to the July 2026 arXiv version of our work~\cite{CDR26}, Kiah, Wei, and Zeng~\cite[Sections 6.5 and 8.2]{KWZ26} showed that our analysis of list-decoding from insertions (namely, our \cref{prop:X-subseq-y}) can also be used to obtain improved rate lower bounds for list-decoding from insertions with fixed list size.

\subsection{AI usage declaration}

GPT-6 obtained the proof of the strong converse for the $\BDC_\delta$ in \cref{thm:strong-converse-BDC} combining Kozlov's argument~\cite{Koz71} with Arimoto's strong converse for DMCs~\cite{Ari73}, initially without mentioning Kozlov's ideas.
A subsequent literature search aided by GPT-6 identified Kozlov's contribution, and the authors rewrote GPT-6's proof to improve exposition.
GPT-5.5 and GPT-6 were also used to revise human-generated writing in some parts of the paper.
With the exception of the proof of \cref{thm:strong-converse-BDC}, all other ideas, results, and proofs in this work are due to the authors.

\section{Preliminaries}

\subsection{Basic notation}

Sets and random variables are both often written using uppercase roman letters, and are easy to distinguish given the context.
We write $\log$ for the base-$2$ logarithm and $\ln$ for the natural logarithm.
For an integer $n\geq 1$, we define $[n]=\{1,2,\dots,n\}$.

A \emph{substring} of a string $x\in\bits^n$ is obtained by taking consecutive symbols from $x$. More precisely, $y\in\bits^\ell$ is a substring of $x\in\bits^n$ if there exists an index $i$ such that $y=x_i x_{i+1}\dots x_{i+\ell-1}$.
A \emph{subsequence} of a string of $x$ is obtained by deleting some (possibly none) of the symbols in $x$.
More precisely, $y\in\bits^\ell$ is a subsequence of $x\in\bits^n$ if there exist indices $1\leq i_1<i_2<\cdots<i_\ell\leq n$ such that $y=x_{i_1} x_{i_2}\dots x_{i_\ell}$.
We also say that $x$ is a \emph{supersequence} of $y$.
Note that every substring is also a subsequence, but not every subsequence is a substring.
We use the notation $y \preceq x$ to mean that $y$ is a subsequence of $x$ (and, equivalently, that $x$ is a supersequence of $y$).

A \emph{run} in a string $x$ is a single-symbol substring of $x$ of maximal length. Every string can be uniquely written as the concatenation of its runs.  
For example, if $x=0111001$, then $x$ is decomposed into $0 \circ 111\circ 00 \circ 1$ where the symbol $\circ$ denotes concatenation of two strings. For a string $x$, we denote by $r(x)$ the number of runs of $x$.

\subsection{Deletion and insertion balls}

We define the \emph{$t$-deletion ball} centered at $y\in\bits^{n}$ by 
\begin{equation*}
    \Bdel_{t}(y)=\{x\in\bits^{|y|-t}: x \preceq y\}.
\end{equation*}
The central challenge in establishing rate vs.\ error-correction trade-offs for insertion- and deletion-correcting codes lies in the irregular size of deletion balls, which varies with the specific structure of the center string. 
Levenshtein~\cite{Lev66} showed that the size of the deletion ball of a string $x$ can be both upper and lower bounded in terms of the number of runs in $x$.
Improved lower bounds were obtained by Calabi and Hartnett~\cite{CH69}, Hirschberg and Régnier~\cite{HR02}, and Liron and Langberg~\cite{LL15}. 
For our result on the capacity of list-decoding from deletions the following bound of Levenshtein will suffice.
\begin{lemma}[{\cite[Equation (1)]{Lev66}}] \label{lem:lev-bound}
    For every $y\in \bits^n$ it holds that
    \[
    \binom{r(y) - t + 1}{t}\leq |\Bdel_{t}(y)| \leq \binom{r(y) + t - 1}{t}.
    \]
\end{lemma}

For $y\in\bits^{n}$, we define its \emph{$t$-insertion ball} as
\begin{equation*}
    \Bins_{t}(y) = \{x\in\{0,1\}^{|y|+t}: y \preceq x\}.
\end{equation*}
The size of the insertion ball around $y$ does \emph{not} depend on $y$. 
We have the following lemma due to Chvátal and Sankoff~\cite{CS75}.
\begin{lemma}[{\cite[Lemma 1]{CS75}}]\label{lem:ins-ball}
    For any $y\in \zo^n$,
    \[
    | \Bins_{t} (y)| = \sum_{i=0}^t \binom{n + t}{i}. 
    \]
\end{lemma}

We now formally define list-decoding from insertions and list-decoding from deletions.
\begin{definition}[List-decoding from insertions and list-decoding from deletions] \label{def:list-decodable-insdel}
    A code $\cC \subseteq \bits^n$ is \emph{$(\delta, L)$-list-decodable from insertions} if for  every $y\in \bits^{(1+\delta)n}$ we have that\footnote{For the sake of simplicity, we assume that $\delta n$ is an integer throughout the paper. Taking floors and ceilings does not change our results.}
    \[
    \left| \Bdel_{\delta n}(y)  \cap \cC \right| \leq L.
    \]

    Analogously, we say that a code $\cC \subseteq \bits^n$ is \emph{$(\delta, L)$-list-decodable from deletions} if for every $y\in \bits^{(1-\delta)n}$ we have that 
    \[
        \left| \Bins_{\delta n}(y)  \cap \cC \right| \leq L.
    \]
\end{definition}

\begin{remark}
    In the standard definition of list-decoding a $(\delta, L)$-list-decodable code is required to handle any \emph{up to} $\delta n$ errors. 
    For insertions and deletions we may consider exactly $\delta n$ errors, as we do in \cref{def:list-decodable-insdel}, without loss of generality.
    Indeed, if for some $y$ with $|y|\geq (1-\delta)n$ we have $|\Bdel_{n-|y|}(y)\cap \cC|>L$, then the same holds for any length-$(1-\delta)n$ subsequence of $y$.
    An analogous argument holds for insertions.
\end{remark}

\subsection{Large deviations}

Some basic results from the theory of large deviations will be useful in our study of the limitations of codes generated by order-$1$ Markov chains 
for list-decoding from deletions in \cref{sec:Cdel-markov}.
We collect them here.

    \begin{lemma}[{Cramér's theorem~\cite[Theorem 2.2.3]{DZ09}}]\label{thm:cramer}
        Let $(Z_i)_{i\in\N}$ be a sequence of i.i.d.\ real-valued random variables with common moment-generating function $M(\gamma)=\E[e^{\gamma Z_1}]$.
        Then, for any open set $U\subseteq\R$ we have
        \begin{equation*}
            \liminf_{n\to\infty}\frac{1}{n} \ln \Pr\left[\frac{1}{n}\sum_{i=1}^n Z_i\in U\right]\geq -\inf_{a\in U}I(a),
        \end{equation*}
        where\footnote{Note that we allow $\gamma$'s such that $M(\gamma)=\infty$.} 
        \begin{equation*}
            I(a) = \sup_{\gamma\in \R}[\gamma a - \ln M(\gamma)]
        \end{equation*}
        is the Lagrange-Fenchel transform of $M$.
    \end{lemma}

The following lemma states a basic property of the Lagrange-Fenchel transform $I$.
It is a consequence of~\cite[Lemma 2.2.5 (b)]{DZ09}.
\begin{lemma}\label{lem:prop-LF}
    Let $Z$ be a random variable with moment-generating function $M$, and let $I$ denote the Lagrange-Fenchel transform of $M$.
    If $M(\gamma),M(\gamma')<\infty$ for some $\gamma<0<\gamma'$, then $a\mapsto I(a)$ is non-increasing in $(-\infty , \E[Z])$, non-decreasing in $(\E[Z],\infty)$, and $I(\E[Z])=0$.
\end{lemma}

\section{Capacity of list-decoding from insertions}\label{sec:Cins}

In this section we prove \cref{thm:Cins-intro}, which we restate here.

\cins*

As mentioned in \cref{sec:intro}, the upper bound on $\Cins(\delta)$ in \cref{thm:Cins-intro} was already established by Haeupler, Shahrasbi, and Sudan~\cite[Theorem 1.2]{HSS18}.
For the sake of completeness we briefly discuss their proof here.
Fix the fraction of insertions $\delta\in[0,1]$ and a code $\cC\subseteq\bits^n$ of rate $R\geq (1+\delta)\left(1-h\left(\frac{\delta}{1+\delta}\right)\right)+\eps$.
By an averaging argument, $\cC$ must have list size
\begin{equation*}
    L\geq 2^{-(1+\delta)n} \sum_{c\in\cC}|\Bins_{\delta n}(c)|.
\end{equation*}
By \cref{lem:ins-ball}, we have $|\Bins_{\delta n}(c)|=\sum_{i=0}^{\delta n} \binom{(1+\delta)n}{i}\geq 2^{\left((1+\delta)h\left(\frac{\delta}{1+\delta}\right)-o(1)\right)n}$.
Plugging this lower bound into the lower bound on $L$ above, we conclude that $L\geq 2^{n\left(R-(1+\delta)\left(1-h\left(\frac{\delta}{1+\delta}\right)\right)-o(1)\right)}\geq 2^{\eps n/2}$ by the lower bound on $R$ above, provided that $n$ is sufficiently large.

Below, we focus on proving a matching lower bound, as stated more precisely in the following theorem.

\begin{theorem} \label{thm:l-d-upb-ins}
    Fix $\delta\in[0,1]$ and $\eps>0$.
    Then, for any $R<(1+\delta)\left(1-h\left(\frac{\delta}{1+\delta}\right)\right)-\eps$ and any sufficiently large $n$ there is a code $\cC\subseteq\bits^n$ of rate $R$ that is $(\delta,L=\Theta(1/\eps))$-list-decodable from insertions.
\end{theorem}

As discussed in \cref{sec:contr-ins}, we prove \cref{thm:l-d-upb-ins} by analyzing a code sampled with memory.
The following lemma is the key behind our result.

\begin{lemma} \label{prop:X-subseq-y}
    Fix $\delta \in [0,1]$ and $\alpha\in (1/2,1)$. 
    Let $X = (X_1, X_2, \ldots, X_n)$ be sampled according to an order-$1$ Markov chain over with each $X_i\in\bits$, $\Pr[X_1=1]=1/2$, and $\Pr[X_{i+1}=X_{i}] = \alpha$.
    Then, for every $y\in\bits^{(1+\delta)n}$ it holds that
    \[
    \Pr[X\preceq y] \leq \left( \frac{\alpha^{1+\delta}}{(2\alpha - 1)^{\delta}} \right) ^n  .
    \]
\end{lemma}
\begin{proof}
    Fix an arbitrary $y\in\bits^{(1+\delta)n}$.
    Following the definition of the $X_i$'s in the lemma statement, we consider the unbounded random process $X_1,X_2,\dots,X_n,X_{n+1},\dots$ and define the following lexicographic matching procedure.
    Starting from $i = 1$, we successively match each bit $X_i \in \bits$ to the leftmost unmatched bit of $y$ that is equal to $X_i$, and stop once $X_i$ cannot be matched to any leftover bit of $y$. 
    Denote by $N(y)$ the number of symbols of $X$ that were matched to $y$ using this procedure.
    Then, we have that
    \begin{equation*}
        \Pr[(X_1,\dots,X_n)\preceq y] = \Pr[N(y)\geq n].
    \end{equation*}
    Additionally, for $b\in \bits$ denote by $N_b(y)$ the number of symbols of $X$ that were matched to $y$ using this procedure, conditioned on $X_1 = b$. 
    Note that
    \begin{equation} \label{eq:Pr-X-superseq-y} 
        \Pr[X\preceq y] = \frac{1}{2}\Pr[N_0(y)\geq n] + \frac{1}{2}\Pr[N_1(y)\geq n],
    \end{equation}
    since $X_1$ is uniform over $\bits$.
    Define the moment-generating functions $M_b(\tau, y) = \E[e^{\tau N_b(y)}]$ and $M(\tau,y) = \frac{1}{2}(M_0(\tau, y) + M_1 (\tau, y))$.
    For any $\tau> 0$ it holds that
    \begin{align} \label{eq:markov-bound}
        \Pr[X\preceq y] &= \frac{1}{2} \cdot \left(\Pr[N_0(y) \geq n] + \Pr[N_1(y) \geq n]\right) \nonumber \\
        &= \frac{1}{2} \cdot \left(\Pr[e^{\tau N_0(y)} \geq e^{\tau n}] + \Pr[e^{\tau N_1(y)} \geq e^{\tau n}]\right) \nonumber \\
        &\leq \frac{1}{2} \left( e^{-\tau n} M_0(\tau, y) + e^{-\tau n} M_1(\tau, y)  \right) \nonumber \\
        &=e^{-\tau n} M(\tau, y),
    \end{align}
    where the inequality follows from Markov's inequality.
    We will carefully choose $\tau$ later to optimize our upper bound on $\Pr[X\preceq y]$ above.

    We now take a closer look at $M(\tau, y)$.
    At a high level, we will study how the moment-generating function behaves as we successively try to match the next bit of $y$ to our currently unmatched bit of $X$.
    Towards that, note that we can extend the notions above as is to $N(z)$, $N_b(z)$, $M(\tau,z)$, and $M_b(\tau,z)$ for any $z$.
    
    Write $y=b\circ z$ for $b\in\bits$ and some bitstring $z$.
    By our matching procedure and the generation of $X$ the random variables $N_0(y)$, $N_1(y)$ and $N_0(z)$, $N_1(z)$ satisfy the relationships described in \cref{tab:rel-yz}.

    \begin{table}[h]
    \centering
    \begin{tabular}{c|c|c}
                & $y = 0\circ z$ & $y = 1 \circ z$\\
                \hline
     $X_1 = 0, X_2 = 0$  &  $N_0(y) = 1 + N_0(z)$ & $N_0(y) = N_0(z)$\\
                \hline
     $X_1 = 0, X_2 = 1$  &  $N_0(y) = 1 + N_1(z)$ & $N_0(y) = N_0(z)$ \\
     \hline
     $X_1 = 1, X_2 = 0$  & $N_1(y) = N_1(z)$  & $N_1(y) = 1 + N_0(z)$ \\
                \hline
     $X_1 = 1, X_2 = 1$  & $N_1(y) = N_1(z)$  &  $N_1(y) = 1 + N_1(z)$ \\
    \end{tabular}
    \caption{Relationship between $N_0(y)$, $N_1(y)$ and $N_0(z)$, $N_1(z)$.}\label{tab:rel-yz}
    \end{table}
    
    Therefore, 
    \begin{align}
        M_{0}(\tau,0\circ z) &= \alpha \E[e^{\tau(1+N_{0}(z))}] +(1-\alpha)\E[e^{\tau(1+N_{1}(z))}] = e^{\tau}\left( \alpha M_{0}(\tau,z) + (1-\alpha)M_{1}(\tau,z) \right), \label{eq:M00}\\
        M_{0}(\tau,1\circ z) &= \E[e^{\tau N_{0}(z)}] = M_0(\tau, z).\label{eq:M01}
    \end{align}
    Similarly, 
    \begin{align}
        M_{1}(\tau,0\circ z) &= M_1(\tau, z), \label{eq:M10}\\
        M_{1}(\tau,1\circ z)
        &=e^{\tau}\left( \alpha M_{1}(\tau,z) + (1-\alpha)M_{0}(\tau,z) \right)  . \label{eq:M11}
    \end{align}
    Combining \cref{eq:M00,eq:M10}, we get 
    \begin{align}
        M(\tau, 0\circ z) &= \frac{1}{2}(M_0(\tau, 0\circ z) + M_1 (\tau, 0\circ z)) \nonumber\\
        & = \frac{e^{\tau} + 1}{2} M_1 (\tau, z) + \frac{\alpha e^{\tau}}{2} (M_0(\tau, z) - M_1(\tau, z)) \nonumber\\
        &= \frac{e^{\tau} + 1}{2} M(\tau, z) + \frac{2\alpha e^{\tau} - e^{\tau} - 1}{2} \cdot \frac{M_0(\tau, z) - M_1(\tau, z)}{2},\label{eq:UB-0z}
    \end{align}
    where the last equality follows by noting that $M_1(\tau,z) = M(\tau, z) - \frac{1}{2}(M_0(\tau,z) - M_1(\tau,z))$. 
    Similarly, from \cref{eq:M01,eq:M11} we get that
    \begin{equation}\label{eq:UB-1z}
        M(\tau, 1\circ z) = \frac{e^{\tau} + 1}{2} M(\tau, z) + \frac{2\alpha e^{\tau} - e^{\tau} - 1}{2} \cdot \frac{M_1(\tau, z) - M_0(\tau, z)}{2} .
    \end{equation}
    Now, observe that, by the triangle inequality,
    \begin{equation*}
        \left|\frac{1}{2}(M_0(\tau, z) - M_1(\tau, z))\right|= \left|\frac{1}{2}(M_1(\tau,z) - M_0(\tau,z))\right| \leq M(\tau, z).
    \end{equation*}
    Therefore, by \cref{eq:UB-0z,eq:UB-1z}, for any $b \in \bits$ we have
    \begin{equation*}
        M(\tau, b\circ z) \leq \frac{e^{\tau} + 1 + |2\alpha e^{\tau} - e^{\tau} - 1|}{2} \cdot M(\tau, z) \leq \max (\alpha e^{\tau}, e^{\tau}(1 - \alpha) + 1) \cdot M(\tau, z).
    \end{equation*}
    Now, if we take $\tau = \ln \frac{1}{2\alpha - 1}>0$ (which is well-defined, recalling that $\alpha\in(1/2,1)$, and minimizes the maximum of the two elements), we get that
    \begin{equation}\label{eq:rec-step}
        M\left(\tau=\ln \frac{1}{2\alpha - 1}, b\circ z\right) \leq \frac{\alpha}{2\alpha - 1} \cdot M(\tau, z).
    \end{equation}
    
    By recursively applying \cref{eq:rec-step} to the $(1+\delta)n$ bits of $y$, we conclude that
    \begin{equation} \label{eq:moment-gen-bound}
        M\left(\ln \frac{1}{2\alpha - 1}, y \right) \leq \left( \frac{\alpha}{2\alpha - 1} \right)^{(1+\delta)n} \cdot M\left( \ln \frac{1}{2\alpha - 1}, \eps \right) \leq \left( \frac{\alpha}{2\alpha - 1} \right)^{(1 + \delta)n},
    \end{equation}
    where the second inequality holds because for the empty string $\epsilon$ there are no symbols to match to, and so $N_b(\eps) = 0$ with probability $1$ and $M (\tau, \eps) = 1$ for all $\tau$. 
    Combining \cref{eq:markov-bound,eq:moment-gen-bound}, we get
    \[
        \Pr[X\preceq y] \leq \left( \frac{1}{2\alpha - 1} \right)^{-n} \cdot \left( \frac{\alpha}{2\alpha - 1} \right) ^{(1 + \delta)n} = \left( \frac{\alpha^{1+\delta}}{(2\alpha - 1)^{\delta}} \right) ^n. \qedhere
    \]
\end{proof}

We now prove \cref{thm:l-d-upb-ins} with the help of \cref{prop:X-subseq-y}.
\begin{proof}[Proof of \cref{thm:l-d-upb-ins}]
    We will assume that $0<\delta<1$, since it is easy to see that $\Cins(0)=1$ and $\Cins(1)=0$.
    Let $\cC\subseteq\bits^n$ be obtained by generating $2^{Rn}$ independent samples from the random process $X$ described in \Cref{prop:X-subseq-y}.
    We will first see $\cC$ as a multiset, and later we will remove repeated codewords to conclude the argument (this will not affect the asymptotic rate).
    Since codewords are generated independently and identically according to the distribution of the order-$1$ Markov chain $X$, the probability that a fixed $y$ is a supersequence of a fixed set of $L+1$ codewords of $C$ is $\Pr[X\preceq y]^{L+1}$.
    Union bounding over all $y\in \bits^{(1+\delta)n}$ and over all subsets of $L+1$ codewords, we get that the probability that $\cC$ is \emph{not} $(\delta, L)$-list-decodable from insertions is at most
    \[
        \sum_{y\in \bits^{(1+\delta)n}} \binom{2^{^{Rn}}}{L+1} \cdot \Pr[X\preceq y]^{L+1}.
    \]
    Invoking \Cref{prop:X-subseq-y}, we get that this quantity is upper bounded by
    \[
        2^{(1 + \delta)n}\cdot 2^{Rn (L+1)} \cdot \left( \frac{\alpha^{1+\delta}}{(2\alpha - 1)^{\delta}} \right)^{n(L+1)} = 2^{n(L+1)\left(R + (1+\delta) \log \alpha - \delta \log(2\alpha - 1) + \frac{1+\delta}{L+1}\right)}.
    \]
    Thus, if 
    \begin{equation}\label{eq:Cins-rate-bound}
        R \leq \delta \log (2 \alpha- 1) - (1+\delta)\log \alpha - \frac{1 + \delta}{L+1} - \frac{\varepsilon}{2},
    \end{equation}
    then the probability that $\cC$ is $(\delta, L)$-list-decodable from insertions is at least $1-2^{-\varepsilon n/2}$. 
    Choose $\alpha = \frac{1+\delta}{2}\in(1/2,1)$ and $L = \frac{2(1 + \delta )}{\varepsilon} - 1$. 
    Then, \cref{eq:Cins-rate-bound} gives that $\cC$ is $(\delta, L = \frac{2(1 + \delta )}{\varepsilon} - 1)$-list-decodable from insertions with probability at least $1-2^{-\varepsilon n/2}$ provided that
    \begin{equation}\label{eq:Cins-rate-bound-opt}
        R \leq (1 +\delta ) + \delta \log\delta -(1+\delta)\log (1+\delta) - \varepsilon =(1+\delta)\left(1 - h\left( \frac{\delta}{1+\delta} \right) \right) - \varepsilon.
    \end{equation}

    Finally, we argue that the number of pairwise distinct codewords in $\cC$ (when seen as a multiset) is at least $2^{Rn-1}$ with high probability over the sampling of $\cC$. 
    Since removing codewords from $\cC$ only improves the list-decoding guarantees, this means that with high probability $C$ will have size at least $2^{Rn-1}$ and will be $(\delta,L)$-list-decodable from insertions.
    Together with \cref{eq:Cins-rate-bound-opt}, this concludes the proof.

    To see the claim about the number of pairwise distinct codewords in $\cC$, first note that the probability that the Markov process generates two identical codewords is
    \begin{equation*}
        \frac{1}{2} (1 - 2\alpha (1 - \alpha))^{n-1} = (1 + \delta^2)^{n-1}2^{-n} \leq 2^{-n(1 - \log(1+\delta^2))},
    \end{equation*}
    where the equality follows by recalling that $\alpha = \frac{1+\delta}{2}$. 
    Consequently, the expected number of pairs of identical codewords in $\cC$ is at most
    \begin{align}
        \binom{2^{Rn}}{2} \cdot 2^{-n(1 - \log(1+\delta^2))} &\leq 2^{n(2R - 1 + \log(1 + \delta^2))} \nonumber\\
        &\leq 2^{n\left(R + \delta \left(1 + \log \frac{\delta}{1+\delta} \right) + \log \frac{1+\delta^2}{1 + \delta} \right)} \nonumber\\
        & = 2^{n (R - c_\delta)},\label{eq:Cins-rep-codewords}
    \end{align}
    where $c_\delta=-\delta \left(1 + \log \frac{\delta}{1+\delta} \right) - \log \frac{1+\delta^2}{1 + \delta}>0$ for all $\delta\in (0,1)$. 
    The second inequality holds because $R < (1+\delta)(1 - h(\frac{\delta}{1+\delta}))$ by \cref{eq:Cins-rate-bound-opt}. 
    The fact that $c_\delta>0$ holds because $\log \frac{\delta}{1+\delta} < -1$ and $\log \frac{1 + \delta^2}{1+\delta} < 0$ for $\delta\in(0,1)$.
    
    Combining Markov's inequality with the upper bound on the expected number of pairs of identical codewords in $\cC$ from \cref{eq:Cins-rep-codewords}, we conclude that there will be at most $2^{n(R-c_\delta/2)}$ pairs of identical codewords in $\cC$ with probability at least $1-2^{-\frac{c_\delta n}{2}}$.
    In particular, this means that for all sufficiently large $n$, the number of pairs of repeated codewords in $\cC$ will be at most $2^{Rn-1}$ with probability at least $1-2^{-\frac{c_\delta n}{2}}$, and so removing repeated codewords from $\cC$ will yield a code $\cC'\subseteq\bits^n$ of size at least $2^{Rn-1}$ with probability at least $1-2^{-\frac{c_\delta n}{2}}$.
    Moreover, if $\cC$ is $(\delta,L)$-list-decodable from insertions then so is $\cC'$.
    We saw above that $\cC$ is $(\delta,L)$-list-decodable from insertions with probability at least $1-2^{-\frac{\eps n}{2}}$.
    Therefore, by a union bound, $\cC'$ will have size at least $2^{Rn-1}$ and will be $(\delta,L)$-list-decodable from insertions simultaneously with probability at least $1-2^{-\frac{c_\delta n}{2}}-2^{-\frac{\eps n}{2}}$, which is positive for all sufficiently large $n$ (recall that $c_\delta<0$ for all $\delta\in(0,1)$). \qedhere
\end{proof}

\section{Capacity of list-decoding from deletions}

\subsection{Order-1 Markov codes do not beat uniformly random coding for deletions}\label{sec:Cdel-markov}

In this section we prove \cref{thm:Cdel-markov}, which we restate here.

\markov*

Fix $\alpha\in[0,1]$.
We prove \cref{thm:Cdel-markov} by identifying an explicit choice of $y\in\bits^{(1-\delta)n}$ such that
\begin{equation*}
    |\Bins_{\delta n}(y)\cap \cC| \geq 2^{\Omega(\eps n)}
\end{equation*}
with high probability over the sampling of $\cC$.
Controlling this intersection amounts to controlling $\Pr[y\preceq X]$, with $X$ sampled as in the theorem statement.
We accomplish this via Cramér's theorem (\cref{thm:cramer}).
Our choice of $y$ depends on whether $\alpha<1/2$ or $\alpha\geq 1/2$, and so we divide the proof into two cases.

\paragraph{The $\alpha<1/2$ case.}
First, suppose that $\alpha<1/2$.
In this case we will show that, with high probability, the alternating string $y=0101\dots 01\in\bits^{(1-\delta)n}$ will have a list of size $2^{\Omega(\eps n)}$ whenever $R>1-h(\delta)+\eps$.
To control this, we are interested in lower bounding $\Pr[y\preceq X]$,
where $X$ is sampled as described in \cref{thm:Cdel-markov}.
In their study of the capacity of the deletion channel, Diggavi and Grossglauser~\cite{DG06} gave a simple moment-based argument for upper bounding this expression.
Their argument, coupled with Cramér's theorem, also implies that the upper bound is asymptotically sharp as $n\to \infty$, as we formalize next.

\begin{lemma}\label{lem:lb-subseq}
    Fix any $\alpha,\delta\in(0,1/2)$.
    Then, for $y=0101\dots01\in\bits^{(1-\delta)n}$ we have
    \begin{equation*}
        \Pr[y\preceq X] \geq e^{-(C^\star(\alpha,\delta)+o(1))n},
    \end{equation*}
    where
    \begin{equation*}
        C^\star(\alpha,\delta)=\begin{cases}
            \ln(\delta/\alpha)-(1-\delta)\ln\left(\frac{\delta(1-\alpha)}{(1-\delta)\alpha}\right), &\textrm{ if $\delta<\alpha$,}\\
            0, &\textrm{ otherwise.}
        \end{cases}
    \end{equation*}
\end{lemma}
\begin{proof}
    We will work under the conditioning that $X_1=0$.
    Since this happens with probability $1/2$, it only affects the $o(1)$ term in the exponent of the lower bound in the lemma statement.

    Let $N_j$ denote the number of bits of the unbounded random process $X_1,X_2,\dots,X_n,X_{n+1},\dots$ required to match the $j$-th bit of $y$ after the $(j-1)$-th bit of $y$ has been matched.
    Note that $N_1=1$ since we assume that $X$ always starts with a $0$.
    Furthermore, $N_2,\dots,N_m$, where $m=(1-\delta)n$, are i.i.d., and
    \begin{equation}\label{eq:sum-char}
        \Pr[y\preceq X] = \Pr[N_2+\cdots+N_m\leq n-1].
    \end{equation}
    Note that
    \begin{equation*}
        \Pr[N_2=i]=\alpha^{i-1}(1-\alpha),
    \end{equation*}
    and so the moment-generating function of $N_2$ is
    \begin{equation*}
        M(\gamma)=\E[e^{\gamma N_2}]=\frac{1-\alpha}{e^{-\gamma}-\alpha}
    \end{equation*}
    whenever $\gamma<\ln(1/\alpha)$, and $M(\gamma)=\infty$ otherwise.
    
    We now combine \cref{thm:cramer} with \cref{eq:sum-char} to conclude that
    \begin{align}
        \Pr[y\preceq X] &= \Pr[N_2+\cdots+N_m\leq n-1]\nonumber\\
        &\geq \Pr\left[N_2+\cdots+N_m\in\left(0, \frac{m-1}{1-\delta}\right)\right]\nonumber\\
        &\geq \exp\left(-m\cdot \inf_{a\in\left(0,\frac{1}{1-\delta}\right)} I(a)-o(m)\right)\nonumber\\
        &=\exp\left(-(1+o(1))n\cdot (1-\delta)\inf_{a\in\left(0,\frac{1}{1-\delta}\right)} I(a)-o(n)\right)\label{eq:LBsubseq}
    \end{align}
    with
    \begin{equation*}
        I(a) = \sup_{\gamma\in \R}\left[\gamma a - \ln M(\gamma)\right]= \sup_{\gamma\in \R}\left[\gamma a - \ln\left(\frac{1-\alpha}{e^{-\gamma}-\alpha}\right)\right].
    \end{equation*}
    
    Then, we leverage \cref{lem:prop-LF}.
    Since $M(\gamma)<\infty$ for all $\gamma<\ln(1/\alpha)$ and $\ln(1/\alpha)>0$, and since $\frac{1}{1-\delta}<\E[N_2]=\frac{1}{1-\alpha}$ if and only if $\delta<\alpha$, we get
    \begin{equation}\label{eq:rate}
        \inf_{a\in\left(0,\frac{1}{1-\delta}\right)} I(a) =\begin{cases}
            I\left(\frac{1}{1-\delta}\right), &\textrm{ if $\delta<\alpha$,}\\
            0, &\textrm{ otherwise.}
        \end{cases}
    \end{equation}

    Finally, it is not hard to see that the expression inside the supremum in $I(\frac{1}{1-\delta})$ when $\delta<\alpha$ is maximized when $\gamma=\ln(\delta/\alpha)$, and algebraic manipulation yields
    \begin{equation}\label{eq:optrate}
        (1-\delta) \cdot I\left(\frac{1}{1-\delta}\right)=\ln(\delta/\alpha)-(1-\delta)\ln\left(\frac{\delta(1-\alpha)}{(1-\delta)\alpha}\right)
    \end{equation}
    in this case.
    Combining \cref{eq:LBsubseq,eq:rate,eq:optrate} concludes the proof.
\end{proof}

We now prove a universal upper bound on $C^\star(\alpha,\delta)$ over all $\alpha\in(0,1/2)$.
\begin{lemma}\label{lem:worsethanfullyrandom}
    Fix any $\alpha,\delta\in(0,1/2)$.
    Let $C^\star(\alpha,\delta)$ be the quantity from \cref{lem:lb-subseq}.
    Then,
    \begin{equation*}
        C^\star(\alpha,\delta)<(1-h(\delta))\cdot \ln 2.
    \end{equation*}
\end{lemma}
\begin{proof}
    The inequality is obvious when $\delta\geq \alpha$, and so we focus on the case $\delta<\alpha$.
    The two expressions coincide when $\alpha=1/2$, and so it suffices to show that
    $C^\star(\alpha,\delta)$ is increasing on $\alpha\in (\delta,1/2)$.
    To see this, note that the derivative of $C^\star(\alpha,\delta)$ with respect to $\alpha$ is
    \begin{equation*}
        \frac{\alpha-\delta}{\alpha(1-\alpha)},
    \end{equation*}
    which is positive when $\alpha>\delta$.
\end{proof}

We use \cref{lem:worsethanfullyrandom} to argue that $\alpha<1/2$ is always a worse choice than $\alpha=1/2$ (a uniformly random code).
\begin{theorem}\label{thm:markov-leq-half}
    Fix a fraction of deletions $\delta\in[0,1/2]$ and $\eps\in(0,h(\delta))$, and let $R=1-h(\delta)+\eps$.
    Suppose that $\cC\subseteq\bits^n$ is obtained by independently sampling $2^{Rn}$ codewords according to an order-$1$ Markov chain $X_1,\dots,X_n$ with $\Pr[X_1=1]=1/2$ and $\Pr[X_i=X_{i-1}]=\alpha$ for some $\alpha\in[0,1/2)$, and keeping repeated codewords.
    Then, for all sufficiently large $n$ and $y=0101\dots 01\in\bits^{(1-\delta)n}$ we have
    \begin{equation*}
    |\Bins_{\delta n}(y)\cap \cC| \geq 2^{\eps n-1}
    \end{equation*}
    with probability at least $1-e^{-2^{\eps n}/8}$ over the sampling of $\cC$.
\end{theorem}
\begin{proof}
    By \cref{lem:lb-subseq,lem:worsethanfullyrandom}, the probability that the alternating string $y=0101\dots 01\in\bits^{(1-\delta)n}$ is a subsequence of $X$ is, for $n$ large enough, at least $2^{-C n}$ for some $C<1-h(\delta)$.
    Let $L$ denote the number of codewords of $\cC$ that are supersequences of $y$.
    Then $\E[L]\geq 2^{Rn}\cdot 2^{-Cn}\geq 2^{\eps n}$.
    Therefore, by a Chernoff bound,
    \begin{equation*}
        \Pr\left[L\leq \frac{1}{2}\cdot 2^{\eps n}\right]\leq e^{-2^{\eps n}/8},
    \end{equation*}
    as desired.
\end{proof}

\paragraph{The $\alpha\geq 1/2$ case.}
In this case we will focus on the all-zeros string $y=00\dots 0\in\bits^{(1-\delta)n}$ and give a lower bound on $\Pr[y\preceq X]$.

\begin{lemma}\label{lem:lb-subseq2}
    Fix any $\alpha\in[1/2,1)$ and $\delta\in(0,1/2)$.
    Then, for $y=00\dots0\in\bits^{(1-\delta)n}$ we have
    \begin{equation*}
        \Pr[y\preceq X] \geq 2^{-(1-h(\delta)+o(1))n}.
    \end{equation*}
\end{lemma}
\begin{proof}
    As in the $\alpha<1/2$ case, we condition on $X_1=0$ and write
\begin{equation*}
    \Pr[y\preceq X] = \Pr[N_2+\cdots+N_m\leq n-1],
\end{equation*}
where $N_j$ is the number of bits from the unbounded random process $X_1,X_2,\dots,X_n,X_{n+1},\dots$ needed to match the $j$-th bit of $y$ after matching the $(j-1)$-th bit of $y$.
The $N_j$'s are i.i.d.\ for $j\geq 2$ with distribution
\begin{equation*}
    \Pr[N_2=i]=\begin{cases}
        \alpha, &\textrm{ if $i=1$,}\\
        (1-\alpha)^2 \alpha^{i-2}, &\textrm{ if $i>1$,}
    \end{cases}
\end{equation*}
and moment-generating function
\begin{equation*}
    M(\gamma)=\E[e^{\gamma N_2}]= \frac{\alpha-e^\gamma(2\alpha-1)}{e^{-\gamma}-\alpha}
\end{equation*}
whenever $\gamma<\ln(1/\alpha)$, and $M(\gamma)=\infty$ otherwise.

Following the proof of \cref{lem:lb-subseq}, since $M(\gamma)<\infty$ for all $\gamma<\ln(1/\alpha)$ and $\ln(1/\alpha)>0$, and since $\frac{1}{1-\delta}<2=\E[N_2]$ because $\delta<1/2$ by hypothesis, we conclude that
    \begin{equation*}
        \Pr[y\preceq X] \geq \exp\left(-(1+o(1))n\cdot (1-\delta) I\left(\frac{1}{1-\delta}\right)\right),
    \end{equation*}
    where
    \begin{equation*}
        (1-\delta)I\left(\frac{1}{1-\delta}\right) =  \sup_{\gamma\in\R}\left[\gamma - (1-\delta)\ln\left(\frac{\alpha-e^\gamma(2\alpha-1)}{e^{-\gamma}-\alpha}\right)\right].
    \end{equation*}
    The derivative of the expression inside the supremum with respect to $\alpha$ is
    \begin{equation*}
        \frac{(1-\delta)(e^\gamma-1)^2}{(1-e^\gamma \alpha)(e^\gamma(2\alpha-1)-\alpha)},
    \end{equation*}
    which is non-positive when $\alpha\in(1/2,1)$, $\delta\in(0,1/2)$, and $\gamma< \ln(1/\alpha)$ (when $\gamma\geq \ln(1/\alpha)$ the expression inside the supremum is always $-\infty$, and so we can disregard it).
    Therefore, for fixed $\gamma< \ln(1/\alpha)$ and $\delta\in(0,1/2)$, the expression inside the supremum is non-increasing in $\alpha\in(1/2,1)$.
    This means that the maximum is achieved at $\alpha=1/2$, in which case the supremum over $\gamma\in \R$ is achieved at $\gamma =\ln(2\delta)$ and
    \begin{equation*}
        (1-\delta)I\left(\frac{1}{1-\delta}\right) = (1-h(\delta))\ln 2.
        \qedhere
    \end{equation*}

\end{proof}

Similarly to the first case, the next result states that $\alpha>1/2$ is always a worse choice than $\alpha=1/2$.
\begin{theorem}\label{thm:markov-geq-half}
    Fix a fraction of deletions $\delta\in[0,1/2]$ and $\eps\in(0,h(\delta))$, and let $R=1-h(\delta)+\eps$.
    Suppose that $\cC\subseteq\bits^n$ is obtained by independently sampling $2^{Rn}$ codewords according to an order-$1$ Markov chain $X_1,\dots,X_n$ with $\Pr[X_1=1]=1/2$ and $\Pr[X_i=X_{i-1}]=\alpha$ for some $\alpha\in(1/2,1)$, and keeping repeated codewords.
    Then, for all sufficiently large $n$ and $y=00\dots 0\in\bits^{(1-\delta)n}$ we have
    \begin{equation*}
    |\Bins_{\delta n}(y)\cap \cC| \geq 2^{\eps n -1}
    \end{equation*}
    with probability at least $1-e^{-2^{\eps n}/8}$ over the sampling of $\cC$.
\end{theorem}
\begin{proof}
    The proof is analogous to that of \cref{thm:markov-leq-half}, but instead we choose the all-zeros string $y=00\dots0\in\bits^{(1-\delta)n}$ and invoke \cref{lem:lb-subseq2} instead of \cref{lem:lb-subseq,lem:worsethanfullyrandom}.
\end{proof}

Combining \cref{thm:markov-leq-half,thm:markov-geq-half} immediately implies \cref{thm:Cdel-markov}.

\subsection{An improved upper bound on the deletion list-decoding capacity}\label{sec:UB-Cdel} 

In this section we prove \cref{thm:capBDC-UB}, which upper bounds $\Cdel(\delta)$ by the capacity $\Capa(\delta)$ of the binary deletion channel with deletion probability $\delta$.
We restate it here.

\BDCUB*

As discussed in \cref{sec:contr-del}, our proof of \cref{thm:capBDC-UB} requires a strong converse for the $\BDC_\delta$ showing that codes of block length $n$ and rate $R>\Capa(\delta)$ must have decoding error probability $1-2^{-\Omega(n)}$.
This strong converse is essentially implicit in Kozlov's work~\cite{Koz71}, with just one ingredient missing (not known at the time).
Roughly speaking, Kozlov's argument proceeds as follows.
Let $(\cC_n)_{n\in\N}$ be a family of codes with each $\cC_n\subseteq\bits^n$, of asymptotic rate $R>\Capa(\delta)$.
First, show that there is a \emph{discrete memoryless channel} (DMC) $\Ch$ with capacity $C<R$ over which it is easier to communicate using $\cC_n$ than over the $\BDC_\delta$.
Second, invoke a strong converse for DMCs to conclude that the decoding error probability of $\cC_n$ over $\Ch$ (and hence the decoding error probability of $\cC_n$ over the $\BDC_\delta$) must be large.
To get the desired strong converse for the $\BDC_\delta$, it then suffices to use Arimoto's later strong converse for DMCs~\cite{Ari73}.

We state the strong converse for the $\BDC_\delta$ implicit in Kozlov's work~\cite{Koz71} below.
For completeness, we also provide a proof.

\begin{lemma}[Implicit in \cite{Koz71}]\label{thm:strong-converse-BDC}
    Fix an arbitrary $\delta\in[0,1]$ and consider an arbitrary family of codes $(\cC_n)_{n\in\N}$ with each $\cC_n\subseteq\bits^n$ and asymptotic rate $R>\Capa(\delta)$.
    Then, for any decoder $\Dec:\bits^{\leq n}\to\bits^n$ we have
    \begin{equation*}
        \max_{c\in\cC_n}\Pr[\Dec(Y_c)\neq c]\geq 1-2^{-\Omega(n)},
    \end{equation*}
    where $Y_c$ denotes the output of $\BDC_\delta$ on input $c\in\cC_n$ and $\Omega(\cdot)$ hides constants depending only on $\delta$.
\end{lemma}

Before we proceed to the proof of \cref{thm:strong-converse-BDC}, we first state here Arimoto's strong converse for general DMCs~\cite{Ari73} and an auxiliary result about the $\BDC_\delta$.
\begin{lemma}[{\cite[Theorems 1 and 2]{Ari73}, adapted}]\label{lem:arimoto}
    Consider an arbitrary discrete memoryless channel $\Ch$ with finite input and output alphabets $\cX$ and $\cY$, respectively.
    Let $\Capa$ denote the Shannon capacity of $\Ch$.
    Consider an arbitrary family of codes $(\cC_n)_{n\in\N}$ with each $\cC_n\subseteq\cX^n$ and asymptotic rate $R>\Capa$.
    Then, for any decoder $\Dec:\cY^n\to\cX^n$ we have
    \begin{equation*}
        \max_{c\in\cC_n}\Pr[\Dec(Y_c)\neq c]\geq 1-2^{-\Omega(n)},
    \end{equation*}
    where $Y_c$ denotes the output of $\Ch$ on input $c\in\cC_n$ and $\Omega(\cdot)$ hides constants depending only on $\Ch$.
\end{lemma}

We will also use the concept of \emph{exact finite-length} deletion channels. 
For a given integer $k\geq 1$ and $\delta\in[0,1]$, consider the DMC $\BDC_{k,\delta}$ with input alphabet $\cX=\bits^k$ and output alphabet $\cY=\bits^{k-\lfloor \delta k\rfloor}$ which given an input $x\in\bits^k$ deletes a uniformly random subset of exactly $\lfloor \delta k\rfloor$ bits from $x$.
We may see $\BDC_{k,\delta}$ as an easier version of $\BDC_\delta$ obtained by fixing the number of bits deleted in each consecutive block of $k$ bits of $x$.
Exact finite-length deletion channels have featured prominently in some prior work on the capacity of $\BDC_\delta$, such as~\cite{FD10,Dal11}.
In particular, we will need the following result.

\begin{lemma}[{\cite{Koz71}, see also \cite[Theorem 1]{Dal11}}]\label{lem:conv-to-capacity}
    Let $\Capa_k(\delta)$ denote the Shannon capacity of $\BDC_{k,\delta}$.
    Then, for every $\delta\in[0,1]$ we have
    \begin{equation*}
        \lim_{k\to\infty} \frac{\Capa_k(\delta)}{k}=\Capa(\delta).
    \end{equation*}
\end{lemma}

We are now ready to prove Kozlov's strong converse.
\begin{proof}[Proof of \cref{thm:strong-converse-BDC}]
    Since $R>\Capa(\delta)$, by \cref{lem:conv-to-capacity} there is a constant $k$ depending only on $\delta$ such that $R>\frac{\Capa_k(\delta)}{k}$.
    For block lengths $n$ of the form $n=mk$ for some integer $m\geq 1$, we can interpret $\cC_n$ as a code $\cC'_m$ in $(\bits^k)^m$ for the $\BDC_{k,\delta}$.
    
    Note that the maximal decoding error probability of $\cC'_m$ over the $\BDC_{k,\delta}$ is at most that of $\cC_n$ over the $\BDC_\delta$, since we can consider decoders for the $\BDC_{k,\delta}$ that ignore partitions between outputs corresponding to consecutive blocks of $k$ bits.
    Furthermore, if $\cC_n$ has rate $R_n$, then $\cC'_m$ has rate $k R_{mk}$, which converges to $kR>\Capa_k(\delta)$ as $m\to\infty$.
    Therefore, \cref{lem:arimoto} gives that the maximal decoding error probability of $\cC'_m$ over the $\BDC_{k,\delta}$ (and hence the maximal decoding error probability of $\cC_n$ over the $\BDC_\delta$) is at least $1-2^{-\Omega(m)}=1-2^{-\Omega(n)}$, recalling that $m=nk$ and that $k$ is a constant depending only on $\delta$.
\end{proof}

\cref{thm:capBDC-UB} follows easily from Kozlov's strong converse in \cref{thm:strong-converse-BDC}.
\begin{proof}[Proof of \cref{thm:capBDC-UB}]
    Let $(\cC_n)_{n\in\N}$ be an arbitrary family of codes $(\cC_n)_{n\in\N}$ with each $\cC_n\subseteq\bits^n$ being $(\delta,L_n)$-list-decodable from deletions, and asymptotic rate $R>\Capa(\delta)$.
    We will show that $L_n=2^{\Omega(n)}$ using \cref{thm:strong-converse-BDC}.

    Consider using $\cC_n$ over the $\BDC_\delta$ with the decoder which given channel output $y\in\bits^{\leq n}$ first checks if $|y|\geq (1-\delta)n$.
    If not, the decoder aborts.
    Otherwise, it constructs the list of up to $L_n$ codewords of $\cC_n$ that are supersequences of $y$, and outputs a uniformly random codeword from this list.
    Since $|y|\geq (1-\delta)n$ with probability $\Omega(1)$, we conclude that the decoder fails with probability at most $1-\Omega(1/L_n)$.
    Since $R>\Capa(\delta)$, by \cref{thm:strong-converse-BDC} we know that the maximal decoding error probability of this decoder is at least $1-2^{-\Omega(n)}$, and so we conclude that $L_n\geq 2^{\Omega(n)}$.
\end{proof}

\subsection{A better upper bound for small fraction of deletions}\label{sec:Cdel-UB-small}

In this section we prove \cref{thm:Cdel-asymp-UB-intro}, which improves on \cref{thm:capBDC-UB} in the low-noise regime.
We restate it here.

\smalldel*

We start with the following simple claim which lower bounds the number of runs most codewords in a code must have given that the code has rate $R$.
\begin{claim} \label{clm:many-runs}
    Fix $\eps > 0$ and $\gamma\in (0,1/2)$, and let $n$ be large enough integer.
    Let $\cC$ be a code of length $n$ with rate $R = h(\gamma) + \eps$. 
    Then, at least half of the codewords in $\cC$ have at least $\gamma n$ runs.
\end{claim}
\begin{proof}
    The total number of binary strings of length $n$ that contain exactly $r$ runs is $2\cdot \binom{n-1}{r-1}$. 
    To see this, note that we have two choices for the initial bit, and then we must select $r-1$ out of the $n-1$ remaining positions at which each new run begins.
    Consequently, the number of binary strings of length $n$ with at most $\gamma n$ runs is 
    \[
        \sum_{i=1}^{\gamma n} 2\cdot \binom{n-1}{i - 1} \leq 2\cdot \gamma n \cdot  \binom{n}{\gamma n} \leq 2^{n \cdot (h(\gamma) + o(1))}  ,
    \]
    where the first inequality follows by noting that $\gamma < 1/2$.
    
    Since $\cC$ has $2^{n(h(\gamma) + \varepsilon)}$ codewords, the number of codewords with at least $\gamma n +1$ runs in $\cC$ is at least 
    \[
    2^{n(h(\gamma) + \varepsilon)} - 2^{n(h(\gamma) + o(1))} = 2^{n (h(\gamma) + \varepsilon)} \left(1 - 2^{n(o(1)-\eps)}\right) \geq \frac{1}{2} \cdot 2^{n (h(\gamma) + \varepsilon)},
    \]
    where the inequality holds for all large enough $n$.
\end{proof}

We are now ready to prove \cref{thm:Cdel-asymp-UB-intro}.

\begin{proof}[Proof of \cref{thm:Cdel-asymp-UB-intro}]
    Fix a fraction of errors $\delta\in(0,1/20)$ and, with some hindsight, set $\gamma=\frac{1}{2}-\sqrt{\frac{h(\delta)\ln 2}{2}}$. 
    Let $\cC$ be an arbitrary code of rate $R$ such that
    \begin{equation} \label{eq:bound-1}
        R \geq \max\left( h(\gamma), 1 - \delta - \left(\gamma - \delta \right) h\left( \frac{\delta}{\gamma - \delta} \right) \right)+ \varepsilon.
    \end{equation}
    We will show that there exists $y\in \bits^{(1-\delta)n}$ such that 
    \[
        \left| \Bins_{\delta n} (y) \cap \cC \right| \geq 2^{(\varepsilon - o(1)) n}  .
    \]
    Since the rate of $\cC$ is at least $h(\gamma) + \varepsilon$, by \cref{clm:many-runs} there is $\cC' \subseteq \cC$ such that $|\cC'| \geq |\cC|/2$ and every codeword in $\cC'$ has at least $\gamma n + 1$ runs. 
    Fix any $c\in \cC'$. 
    By \cref{lem:lev-bound}, the number of subsequences of length $(1 - \delta) n$ of $c$ is at least 
    \[
    \binom{\gamma n - \delta n}{\delta n} \geq 2^{n\cdot (\gamma - \delta) h\left( \frac{\delta}{ \gamma - \delta} \right) - o(n)} ,
    \]
    since $\gamma-\delta>2\delta$ for $\delta < 1/20$.
    Therefore, the number of pairs $(c,y)\in \cC' \times \bits^{n-\delta n}$ such that $y\preceq c$ is at least
    \[
    2^{Rn - 1 + n(\gamma - \delta)h\left( \frac{\delta}{ \gamma - \delta}\right) - o(n)} = 2^{n \left(R +  (\gamma - \delta)h\left( \frac{\delta}{ \gamma - \delta} \right) - o(1) \right)} .
    \]
    Thus, by the pigeonhole principle, there exists $y^{\star}\in \bits ^{(1 - \delta)n}$ such that
    \begin{align*}
        \left| \Bins_{\delta n} (y^{\star}) \cap \cC \right| &= \left| \{ c \in \cC': y^{\star} \preceq c \} \right| \\ 
        &\geq \frac{1}{2^{(1 - \delta)n}} \cdot 2^{n \left(R +  (\gamma - \delta)h\left( \frac{\delta}{ \gamma - \delta} \right) - o(1) \right)}\\
        &= 2^{n \left(R - (1 - \delta) +   (\gamma - \delta)h\left( \frac{\delta}{ \gamma - \delta} \right) - o(1) \right)}\\
        &\geq 2^{n(\varepsilon - o(1))} .
    \end{align*}
    The last inequality follows from \cref{eq:bound-1}.

    We are left to show that when plugging $\gamma=\frac{1}{2}-\sqrt{\frac{h(\delta)\ln 2}{2}}$ in
    \cref{eq:bound-1}, we get the claimed rate. 
    First observe that since $h(1/2 - x)< 1 - 2x^2/\ln2$ for every $x\in (0,1/2)$, we get
    \[
        h(\gamma) = h\left(\frac{1}{2} - \sqrt{\frac{h(\delta)\ln 2}{2}} \right) < 1 - h(\delta)  .
    \]
    Now, suppose that the maximum in \cref{eq:bound-1} is strictly smaller than $1-h(\delta)$. By the argument above, this means that there is $R<1-h(\delta)$ such that all codes of rate $R$ have list size $2^{\Omega(\eps n)}$. 
    But this contradicts \cref{eq:Cdel-bounds}, which states that for any $R<1-h(\delta)$ there is a code of rate $R$ with constant list size. 
    Therefore, the maximum must be at least $1-h(\delta)$, and since the first term inside the maximum is smaller than $1-h(\delta)$, it follows that the maximum equals the second term.
    Therefore, under our choice of $\gamma$, \cref{eq:bound-1} becomes
    \[
        R \geq 1 - \delta - \left(\gamma - \delta \right) h\left( \frac{\delta}{\gamma - \delta} \right) + \varepsilon.
    \]
    Finally, it is not hard to check that
    \begin{equation*}
        1 - \delta - \left(\gamma - \delta \right) h\left( \frac{\delta}{\gamma - \delta} \right) = 1+\delta\log\delta - \frac{\delta}{\ln 2}+o(\delta) = 1-h(\delta)+o(\delta),
    \end{equation*}
    as desired. \qedhere
\end{proof}

\section*{Acknowledgements}

We thank Shu Liu, Ivan Tjuawinata, and Chaoping Xing for a useful and detailed discussion about their prior work~\cite{LTX21}.

D.\ Doron was supported in part by 
Israel Science Foundation grant \#857/25 and by NSF-BSF grant \#2022644.
J.\ Ribeiro was funded by the European Union (LESYNCH, 101218842). Views and opinions expressed are however those of the authors only and do not necessarily reflect those of the European Union or the European Research Council Executive Agency. Neither the European Union nor the granting authority can be held responsible for them.
J.\ Ribeiro was also funded by national funds through FCT – Fundação para a Ciência e a Tecnologia, I.P., and, when eligible, co-funded by EU funds under project/support UID/50008/2025 – Instituto de Telecomunicações, with DOI \href{https://doi.org/10.54499/UID/50008/2025}{10.54499/UID/50008/2025}.

\bibliographystyle{alpha}
\bibliography{refs}

\appendix

\section{Uniformly random coding lower bounds for list-decoding from insertions}\label{app:unif-random}

In this section we characterize the performance of uniformly random codes with respect to list-decoding from insertions, establishing in particular the lower bound in \cref{eq:Cins-LB-prior}.
We stress that this has limited novelty.
Haeupler, Shahrasbi, and Sudan~\cite[Theorem 1.7]{HSS18} prove a lower bound on the list-decoding capacity via uniformly random codes, but appeal to sub-optimal bounds on the size of deletion balls that lead to a worse lower bound.
The main goal of this section is simply to provide a streamlined, optimal analysis of the performance of uniformly random codes in list-decoding from insertions.

By a uniformly random code of rate $R$ we mean a code $\cC\subseteq\bits^n$ obtained by sampling each of its $2^{Rn}$ codewords independently and uniformly at random from $\bits^n$.
Note that there may be repeated codewords, and we see $\cC$ as a multiset.
This makes little difference since with probability at least $1-2^{-\Omega(n)}$ at most a small fraction of codewords of $\cC$ will be repeated, meaning that the effective size of $\cC$ as a set will be at least $2^{Rn-1}$ with high probability, and removing repeated codewords only makes the list sizes in list-decoding smaller.
We note also that other models of uniformly random codes have been considered in the literature, such as sampling $\cC$ as a uniformly random subset of $\bits^n$ of size $2^{Rn}$.
There is no difference between these models when it comes to performance for list-decoding.

The following theorem characterizes the performance of uniformly random codes in list-decoding from a $\delta$-fraction of insertions.
Informally, it states that length-$n$ uniformly random codes of rate smaller than $(1-h(\delta))\cdot \mathbf{1}_{\{\delta<1/2\}}$ have small list size with high probability, while uniformly random codes of larger rate will have exponentially large list size with high probability.
In particular, when $\delta>1/2$ a uniformly random code of positive rate will have exponentially large list size with high probability.
As we already discussed in \cref{sec:intro}, this contradicts some results claimed by Liu, Tjuawinata, and Xing via the analysis of uniformly random codes in list-decoding from insertions~\cite[Corollary 13 and Figure 2]{LTX21}.

\begin{theorem}\label{thm:random-LD-ins}
    Fix $\delta\in[0,1]$ and $\eps>0$.
    Then, a uniformly random code $\cC\subseteq\bits^n$ of rate $R<(1-h(\delta))\cdot \mathbf{1}_{\{\delta<1/2\}}-\eps$ will be $(\delta,L=\frac{2+\delta}{\eps})$-list-decodable from insertions with probability at least $1-2^{-\Omega(\eps n)}$.

    Furthermore, a uniformly random code $\cC\subseteq\bits^n$ of rate $R\geq (1-h(\delta))\cdot \mathbf{1}_{\{\delta<1/2\}}+\eps$ will not be $(\delta,L)$-list-decodable from insertions for any $L\leq 2^{\eps n/3}$ with probability at least $1-2^{-2^{\Omega(\eps n)}}$.
\end{theorem}

Intuitively, proving \cref{thm:random-LD-ins} boils down to controlling the size of the largest \emph{deletion ball} over strings of a given length.
The following result, due to Hirschberg and R\'egnier~\cite{HR02}, building on prior work of Calabi and Hartnett~\cite{CH69}, allows us to do that (see also the discussion in the introduction of~\cite{LL15}).
\begin{lemma}[{\cite[Corollary 2.7]{HR02}}]\label{lem:del-ball}
    For any $y\in\bits^m$ and number of deletions $t$ we have
    \begin{equation*}
        |\Bdel_{t}(y)|\leq \sum_{i=0}^t \binom{m-t}{i},
    \end{equation*}
    and equality is achieved when $y$ is an alternating string.\footnote{A string $y\in\bits^m$ is \emph{alternating} if it has $m$ runs. In particular, if $m$ is even and the alphabet is binary, then $y=0101\dots 01$ or $y=1010\dots10$. This lemma extends to larger alphabets too.}
    Recall that $\Bdel_{t}(y)$ is the set of all length-$(|y|-t)$ subsequences of $y$. 
\end{lemma}

\begin{proof}[Proof of \cref{thm:random-LD-ins}]
    To see the first part of the theorem statement we can focus on $\delta<1/2$.
    By taking a union bound over all the $\binom{2^{Rn}}{L+1}\leq 2^{Rn(L+1)}$ possible choices of $L+1$ codewords and strings $y$ of length $(1+\delta)n$, we get that the probability that $\cC$ is not $(\delta,L)$-list-decodable from insertions is at most
    \begin{equation*}
        2^{Rn(L+1)}\sum_{y\in\bits^{(1+\delta)n}} \left(\frac{|\Bdel_{\delta n}(y)|}{2^n}\right)^{L+1}.
    \end{equation*}
    By \cref{lem:del-ball}, for $\delta< 1/2$ we can upper bound this probability by
    \begin{align*}
        2^{Rn(L+1)} \sum_{y\in\bits^{(1+\delta)n}} \left(2^{-n}\sum_{i=0}^{\delta n} \binom{n}{i}\right)^{L+1} &= 2^{Rn(L+1)} \cdot 2^{(1+\delta)n} \left(2^{-n}\sum_{i=0}^{\delta n} \binom{n}{i}\right)^{L+1}\\
        &\leq 2^{Rn(L+1)} \cdot 2^{(1+\delta)n} \cdot 2^{-n (1-h(\delta))(L+1)},
    \end{align*}
    and the final upper bound is smaller than $1$ when $R=1-h(\delta)-\eps$ and $L= \frac{2+\delta}{\eps}$.

    To see the second part of the theorem statement, fix $\eps>0$ and let $\cC$ be a uniformly random code of rate $R=(1-h(\delta))\cdot \mathbf{1}_{\{\delta<1/2\}}+\eps$.
    By \cref{lem:del-ball}, an alternating string $y$ of length $(1+\delta)n$ has exactly
    \begin{equation*}
        \sum_{i=0}^{\delta n}\binom{n}{i}
    \end{equation*}
    length-$n$ subsequences.
    This number is at least $2^{(1-o(1))n h(\delta)}$ when $\delta<1/2$ and at least $2^{n-1}$ when $\delta\geq 1/2$.
    Therefore, the expected number of subsequences of $y$ that become codewords of $\cC$ is at least $2^{\eps n/2}$ for all sufficiently large $n$, and a Chernoff bound gives that this number will be at least $2^{\eps n/3}$ with probability at least $1-2^{-2^{\Omega(\eps n)}}$.
\end{proof}

\end{document}